\documentclass[aps,onecolumn,a4paper,superscriptaddress,longbibliography,nofootinbib,notitlepage,groupedaddress]{revtex4-1}

\usepackage{lipsum}
\usepackage{setspace}

\usepackage{graphicx}
\usepackage{amsmath}
\usepackage{amssymb}
\usepackage{amsthm}
\usepackage[utf8]{inputenc}
\usepackage{comment}

\usepackage{color}
\usepackage[dvipsnames]{xcolor}

\newcommand{\norm}[1]{\left\lVert#1\right\rVert} % norm: double vertical bars

\usepackage[a4paper,total={6.5in,9in}]{geometry}
\usepackage[english]{babel}
\usepackage[T1]{fontenc}
\usepackage{mathrsfs}
\usepackage{braket}
\usepackage{physics}
\usepackage{enumerate}
\usepackage{graphicx}
\usepackage{mathtools}
\usepackage[export]{adjustbox}
\usepackage{microtype}

\theoremstyle{plain}

\usepackage[colorlinks = true,
            linkcolor = red,
            urlcolor  = blue,
            citecolor = blue,
            anchorcolor = red]{hyperref}

\newtheorem{theorem}{Theorem}

\newtheorem{lemma}[theorem]{Lemma}

\newtheorem*{remark*}{Remark}   %  \newtheorem*{remark}{Remark}  * is for numbering
\renewcommand\qedsymbol{$\blacksquare$}
\newenvironment{proof-of}[1][{\hspace{-\blank}}]{{\medskip\noindent\textit{Proof~{#1}.\ }}}{\hfill\qedsymbol}

\renewcommand{\Tr}{{\operatorname{Tr}\,}}

\newcommand{\1}{\openone}
\newcommand{\proj}[1]{|#1\rangle\!\langle #1|}
\newcommand{\cD}{{\mathcal{D}}}

\newcommand{\nc}{\newcommand}
\nc{\rnc}{\renewcommand}
\nc{\avg}[1]{\langle#1\rangle}
\nc{\Rank}{\operatorname{Rank}}
\nc{\smfrac}[2]{\mbox{$\frac{#1}{#2}$}}

\nc{\ox}{\otimes}
\nc{\dg}{\dagger}
\nc{\dn}{\downarrow}
\nc{\cA}{{\cal A}}
\nc{\cB}{{\cal B}}
\nc{\cC}{{\cal C}}
\nc{\cF}{{\cal F}}
\nc{\cG}{{\cal G}}
\nc{\cH}{{\cal H}}
\nc{\cI}{{\cal I}}
\nc{\cJ}{{\cal J}}
\nc{\cK}{{\cal K}}
\nc{\cL}{{\cal L}}
\nc{\cM}{{\cal M}}
\nc{\cN}{{\cal N}}
\nc{\cO}{{\cal O}}
\nc{\cP}{{\cal P}}
\nc{\cQ}{{\cal Q}}
\nc{\cR}{{\cal R}}
\nc{\cS}{{\cal S}}
\nc{\cX}{{\cal X}}
\nc{\cY}{{\cal Y}}
\nc{\cW}{{\cal W}}
\nc{\cZ}{{\cal Z}}
\nc{\csupp}{{\operatorname{csupp}}}
\nc{\qsupp}{{\operatorname{qsupp}}}
\nc{\RR}{{{\mathbb R}}}
\nc{\CC}{{{\mathbb C}}}
\nc{\FF}{{{\mathbb F}}}
\nc{\NN}{{{\mathbb N}}}
\nc{\ZZ}{{{\mathbb Z}}}
\nc{\PP}{{{\mathbb P}}}
\nc{\QQ}{{{\mathbb Q}}}
\nc{\UU}{{{\mathbb U}}}
\nc{\EE}{{{\mathbb E}}}
\begin{document}
\title{On Strong Converse Bounds for the Private and Quantum \protect\\ Capacities of Anti-degradable
Channels}

\author{Zahra Baghali Khanian}
\email{zbkhanian@gmail.com }
\email{zkhanian@perimeterinstitute.ca}
\affiliation{Perimeter Institute for Theoretical Physics, Ontario, Canada, N2L 2Y5}
\affiliation{Institute for Quantum Computing, University of Waterloo, Ontario, Canada, N2L 3G1}

\author{Christoph Hirche}
\email{christoph.hirche@gmail.com}
\affiliation{Institute for Information Processing (tnt/L3S), Leibniz Universit\"at Hannover, Germany}

\begin{abstract}
 %In this paper we prove that a strong converse bound for the private classical capacity of anti-degradable channels exists. Namely, we show that the capacity is 0 for any error $\epsilon>0$ and privacy $\delta >0$ satisfying $\delta \sqrt{1-\epsilon^2}+\epsilon \sqrt{1-\delta^2}<1$. Moreover, we provide a ``pretty simple'' proof for the ``pretty strong'' converse of the quantum capacity of anti-degradable channels for any error $\epsilon <\frac{1}{\sqrt{2}}$.  
 
We establish a strong converse bound for the private classical capacity of anti-degradable quantum channels. Specifically, we prove that this capacity is zero whenever the error $\epsilon > 0$ and privacy parameter $\delta > 0$ satisfy the inequality $\delta \sqrt{1-\epsilon^2}+\epsilon \sqrt{1-\delta^2}<1$. This result strengthens previous understandings by sharply defining the boundary beyond which reliable and private communication is impossible. Furthermore, we present a ``pretty simple'' proof of the ``pretty strong'' converse for the quantum capacity of anti-degradable channels, valid for any error $\epsilon < \frac{1}{\sqrt{2}}$. Our approach offers  
clarity and technical simplicity, shedding new light on the fundamental limits of quantum communication.

%  In this  paper we prove:\\
%  \noindent 1. A strong converse bound for the private capacity of anti-degradable channels exists. Namely, we show that the capacity is 0 for any $\epsilon$ and $\delta$ satisfying $\delta \sqrt{1-\epsilon^2}+\epsilon \sqrt{1-\delta^2}<1$.
  % (This is stronger than $\epsilon+\delta<1$, which I found in similar classical papers)  \\
% \noindent 2. We provide a ``pretty'' simple proof for the ``pretty'' strong converse of the quantum capacity for anti-degradeable channels for $\epsilon <\frac{1}{\sqrt{2}}$. 
\end{abstract}

\maketitle

\section{Introduction}
%pretty paper intro

The private classical capacity of a quantum channel is the maximum rate at
which classical information can be transmitted reliably and privately, ensuring
that an eavesdropper has no knowledge of the transmitted information. 
On the other hand, the quantum capacity of a channel is defined as the
highest rate at which quantum information can be transmitted reliably.
Both capacities are defined in the limit as the number of uses of the channel
tends to infinity \cite{Shor_direct_capacity2002,Lloyd_capacity_97,Devetak-capacity-2005}.

The quantum capacity theorem consists of two parts: a direct part and a converse part.
The direct part asserts that for rates below a certain threshold, it is possible to
construct codes with decoding errors (and privacy for the private capacity)
that tend to zero as the number of channel uses increases.
The (weak) converse states that if the rate lies above
this threshold then the error does not go to $0$ for any sequence
of codes. 
A strong converse, on the other hand, is the statement that for rates above the capacity the
error converges to its maximum $1$  in the limit of a large number of channel uses.

Except for certain examples of channels (such as PPT entanglement binding channels and the identity channel),  a strong converse is not known in general
for the quantum and private classical capacities.
%Nevertheless, the concept of a ``pretty strong'' converse is introduced in \cite{AW_pretty} as the statement that for rates above the capacity the error  makes a discontinuous jump from 0 to a constant number but not the maximum 1. 
%
Nevertheless, the concept of a ``pretty strong'' converse is introduced in \cite{AW_pretty}. 
They show that for rates above the capacity, the error undergoes a discontinuous jump 
from 0 to a constant value, but not necessarily to the maximum value of 1.
More specifically, for the private capacity of degradable channels they show that for any error $\epsilon>0$
and privacy $\delta>0$ satisfying $\epsilon+2\delta<\frac{1}{\sqrt{2}}$ a 
 pretty strong converse bound holds.
Also, for the quantum capacity of degradable channels, a pretty strong converse bound holds for any $\epsilon <\frac{1}{\sqrt{2}}$.

%It was proved by Morgan and Winter in \cite{AW_pretty} for degradable channels stating that for any error $\epsilon>0$ and privacy $\delta>0$ satisfying $\epsilon+2\delta<\frac{1}{\sqrt{2}}$ the converse bound holds.

As mentioned above, so far only a \textit{pretty} strong converse bound is known for the private classical capacity. 
In this paper, we prove that  a \textit{strong} converse bound holds for the private classical capacity of anti-degradable channels: for any error $\epsilon>0$ and privacy $\delta >0$ satisfying $\delta \sqrt{1-\epsilon^2}+\epsilon \sqrt{1-\delta^2}<1$  
the size of the message set is bounded above by a constant independent of $n$, i.e. the number of channel uses. In other words,
the rate is bounded above by 0 in the limit of $n$ converging to $\infty$.
This bound can be simplified to $\delta +\epsilon <1$. 
So, for example, as long as the privacy $\delta$ is very small, the error $\epsilon$ can be arbitrarily close to 1 (or vice versa).
%
%For the quantum capacity, though, only {pretty} strong converse bounds are known. 
%For degradable channels, it is shown in \cite{AW_pretty} that for any error $\epsilon < \frac{1}{\sqrt{2}}$ the converse bound holds. 
%
%
For the quantum capacity of anti-degradable channels, it is shown in \cite{Kaur_pretty} that a pretty strong converse holds for any error $\epsilon < \frac{1}{2}$.
We provide a slightly stronger result and a simpler proof: a pretty strong converse holds for any error $\epsilon < \frac{1}{\sqrt{2}}$.

%which is established in \cite{Watanabe_p_converse} for  the secret key capacity of classical channels. 
%H. Tyagi and S. Watanabe, “A bound for multiparty secret key agreement and implications for a problem of secure computing,” in Proc. 33rd Annu. Int. Conf. EUROCRYPT, 2014, pp. 369–386.

\bigskip
The paper is structured as follows: In Section~\ref{Sec:Preliminaries} we introduce the necessary notation and preliminary results. In Section~\ref{Sec:Capacities} we define the quantum and private capacities. Then in Sections~\ref{Sec:StrongConversePrivate} and~\ref{Sec:StrongConverseQuantum} we prove our results regarding the private and quantum capacities, respectively. Finally, we discuss our results in Section~\ref{Sec:Discussion}.

\section{Notation and Preliminaries}\label{Sec:Preliminaries}
In this paper, quantum systems are associated with finite dimensional Hilbert spaces $A$, $B$,\dots $R$, etc., whose dimensions are denoted by $|A|$, $|B|$,\dots $|R|$, respectively.
Quantum channels $\cN: \cL(A') \to \cL(B)$, and therefore also the encoding and decoding of quantum information,  are completely positive and trace-preserving (CPTP) maps. 
As there is no  risk of confusion, we denote a map with $\cN: A' \to B$.
We denote the set of normalized and sub-normalized (trace less than or equal to 1)  quantum states on Hilbert space $A$ by $\cS(A)$ and $\cS_{\leq}(A)$, respectively. 
%
%and its inverse function $\exp$, unless otherwise stated, is also to basis $2$).
%
The conditional entropy and the coherent information, $S(A|B)_{\rho}$ and $I(A \rangle B)_{\rho}$, respectively, are defined as
\begin{align*}
  S(A|B)_{\rho}   &= S(AB)_\rho-S(B)_{\rho}, \text{ and} \\ 
  I(A \rangle B)_{\rho} &= -S(A|B)_{\rho} = S(B)_{\rho} - S({AB})_{\rho},
  %I(A:B|C)_{\rho} &= S(A|C)_\rho-S(A|BC)_{\rho} \\
   %                &= S(AC)_\rho+S(BC)_\rho-S(ABC)_\rho-S(C)_\rho.
\end{align*}
where $S(A)_{\rho} = -\Tr \rho \log \rho$ is the von Neumann entropy. 

The fidelity between two states $\rho$ and $\xi$ is defined as 
\(
 F(\rho, \xi) = \|\sqrt{\rho}\sqrt{\xi}\|_1 
                = \Tr \sqrt{\rho^{\frac{1}{2}} \xi \rho^{\frac{1}{2}}},
\) 
with the trace norm $\|X\|_1 = \Tr|X| = \Tr\sqrt{X^\dagger X}$. The purified distance between two states $\rho$ and $\xi$
is defined as $P(\rho, \xi):=\sqrt{1-F^2(\rho, \xi)}$.
It relates to the trace distance in the following well-known way \cite{Fuchs1999}:
\begin{equation}\label{lemma: Fuchs}
  1-F(\rho,\xi) \leq \frac12\|\rho-\xi\|_1 \leq \sqrt{1-F(\rho,\xi)^2}.
\end{equation}

For a bipartite state $\rho^{AB} \in \cS_{\leq}(AB)$ the min-entropy of $A$ conditioned on $B$ is defined as (see \cite{Tomamichel_book})
\begin{align}
H_{\min}(A|B)_{\rho}:=\max_{\sigma_B \in \cS(B)} \max \{\lambda \in \mathbb{R} : \rho^{AB} \leq 2^{- \lambda} \1 \ox \sigma^{B} \}.
\end{align}
With a purification $\ket{\psi}^{ABC}$ of $\rho^{AB}$, we define the max-entropy 
\begin{align}
H_{\max}(A|B)_{\rho}:=-H_{\min}(A|C)_{\psi^{AC}},
\end{align}
with the reduced state $\psi^{AC}=\Tr_B ({\psi}^{ABC})$.

Let $\epsilon \geq 0$
and $\rho^{AB} \in \cS_{\leq}(AB)$.
The $\epsilon$-smooth min-entropy of $A$ conditioned on $B$ is defined as
\begin{align}
H^{\epsilon}_{\min}(A|B)_{\rho}:= \max_{\rho'  \approx_{\epsilon} \rho} H_{\min}(A|B)_{\rho'},
\end{align}
where $\rho' \approx_{\epsilon} \rho$ means $P(\rho',\rho) \leq \epsilon$ for $\rho' \in \cS_{\leq}(AB)$. 
Similarly,
\begin{align}\label{eq: smooth duality}
H^{\epsilon}_{\max}(A|B)_{\rho}&:= \min_{\rho'  \approx_{\epsilon} \rho} H_{\max} (A|B)_{\rho'} \nonumber\\
&=-H^{\epsilon}_{\min}(A|C)_{\psi},
\end{align}
with a purification $\ket{\psi}^{ABC}$ of $\rho^{AB}$ \cite{Dual_min_max_2010}.

The following relation between smooth min-max entropies holds.
\begin{lemma}\label{lemma: H_min < H_max+terms} (Proposition 5.5 in \cite{Tomamichel_PhD})
    Let $\rho \in \cS(AB)$, $\alpha,\beta \geq 0$ and $\alpha+\beta  <\frac{\pi}{2}$. Then
    \begin{align}
        H_{\min}^{\sin(\alpha)}(A|B)_{\rho} &\leq H_{\max}^{\sin(\beta)}(A|B)_{\rho}+\log \frac{1}{\cos^2(\alpha+\beta)}.  
            \end{align}
\end{lemma}

\medskip
A channel $\cN: A \to B$ is called degradable if it can be degraded to its
complementary channel, i.e. if there exists a CPTP  map $\cM$ such
that $\cN_c =\cM \circ \cN$. Introducing the Stinespring dilation of $\cM$
by an isometry $V : B  \hookrightarrow B' E'$, the channel output system
$B$ can be mapped to the composite system $B' \ox E'$ such that
the channel taking $A$ to $E'$ is the same as the channel taking
$A$ to $E$ (with an isomorphism between $E$ and $E'$ fixed once
and for all). 
%We may also assume $B'$ to be minimal.
If the complementary channel is degradable, i.e. if $\cN =\cM \circ \cN_c$ for some CPTP map,
we call $\cN$ anti-degradable. 
%A channel that is both degradable and anti-degradable is called symmetric \cite{}.

\section{Capacity Definitions}\label{Sec:Capacities}
In this section, we define the quantum capacity and the private classical capacity of quantum channels.

\bigskip

\noindent \textbf{The entanglement generation code: }
An $(n,\epsilon)$ entanglement generation code is defined as follows. Alice prepares w.l.o.g a bipartite pure state  $\ket{\psi}^{{A'}^nR}$, and sends ${A'}^n$ to Bob through
$n$ uses of the channel $\cN: A' \to B$ with an isometric extension $U_{\cN}: A' \hookrightarrow BE$. The output state of the isometric extension $U_{\cN}^{\ox n}$ of the channel is
\begin{align}
    \ket{\sigma}^{B^nE^nR}:=(U_{\cN}^{\ox n}\ox \1_{R})\ket{\psi}^{{A'}^nR},
\end{align}
where $E^n$ is the environment system of the channel.
Then, Bob applies the decoding map $\cD_n: B^n \to A$. The output state of the isometric extension $U_{\cD_n}: B^n \hookrightarrow A W$ of the decoding map is
\begin{align}
    \ket{\xi}^{AE^nWR}:=(U_{\cD_n} \ox \1_{E^nR})\ket{\sigma}^{B^nE^nR},
\end{align}
where $W$ is the environment system of the decoding isometry. 
The code has \textit{error} $\epsilon$ if
\begin{align}
    P(\proj{\Phi}^{AR},\xi^{AR}) \leq \epsilon, %=\sqrt{1-F^2(\proj{\Phi}^{QR},\xi^{QR})}=\epsilon,
\end{align}
%
%The fidelity  for the entanglement generation protocol is defined as
%\begin{align}
%    F(\proj{\Phi}^{QR},\xi^{QR})=\sqrt{\bra{\Phi}\xi^{QR}\ket{\Phi}^{QR}},
%\end{align}
where ${\xi}^{AR}=\Tr_{E^n W} (\proj{\xi}^{AE^nWR})$, and ${\Phi}^{AR}$ is a maximally entangled state of dimension $|A|=|R|$.
The rate of the code is defined as $\frac{\log |A|}{n}$.
The maximum dimension $|A|$ such that there exists a
quantum code for $\cN^{\ox n}$ with error $\epsilon$ is denoted by $N(n, \epsilon)$. 
The (asymptotic) quantum capacity $\cQ(\cN)$ of a quantum channel $\cN$ is given by:
\begin{align}
    \cQ(\cN) = \inf_{\varepsilon > 0} \liminf_{n \to \infty} \frac{1}{n} \log N(n, \varepsilon).
\end{align}

\noindent \textbf{The private classical capacity code: }
An $(n,\epsilon,\delta)$ private classical capacity  code is defined as follows. Alice prepares the state below,
\begin{align}\label{eq: ideal rho^XX'}
    \rho^{XX'} = \sum_m \frac{1}{M} \proj{m}^X \otimes \proj{m}^{X'}.
\end{align}
Then, she applies an encoding map $\mathcal{E}_n: X \to {A'}^n$ to prepare
\begin{align}
    \nu^{{A'}^n} = \sum_m \frac{1}{M} \nu_m^{{A'}^n} \otimes \proj{m}^{X'}.
\end{align}
She sends ${A'}^n$ to Bob through
$n$ uses of a channel $\cN: A' \to B$ with an isometric extension $U_{\cN}: A' \hookrightarrow BE$. The output state of the isometric extension $U_{\cN}^{\ox n}$ of the channel is
\begin{align}
    \sigma^{B^n E^n X'} = \sum_m \frac{1}{M} \sigma_m^{B^n E^n} \otimes \proj{m}^{X'}.
\end{align}
Bob receives $B^n$ and applies the decoding map $\cD_n: B^n \to X$ 
\begin{align}
    \xi^{{X} E^n X'} = \sum_m \frac{1}{M} \xi_m^{{X} E^n} \otimes \proj{m}^{X'}.
\end{align}
The  code  has respectively \textit{error} $\epsilon$ and \textit{privacy} $\delta$ if
\begin{align}
    P\left( \xi^{{X} X'}, \rho^{X  X'} \right) &\leq \epsilon, \label{eq: P epsilon} \\
     P\left(\xi^{E^n X'}, \rho^{E^n  X'} \right) &\leq \delta \label{eq: P delta},
\end{align}
where $\rho^{ E^n X'} = \omega^{E^n} \ox \sum_m \frac{1}{M} \proj{m}^{X'}  $, and
$\omega^{E^n}$ is a fixed state.
For a given channel $\cN$, we denote the largest $M$ such that there exists a private classical capacity  code with error $\epsilon$
and privacy $\delta$ by $M(n,\epsilon,\delta)$. 
The (asymptotic) private classical capacity $\cP(\cN)$ of a quantum channel $\cN$ is given by:
\begin{align}
    \cP(\cN) = \inf_{\varepsilon,\delta > 0} \liminf_{n \to \infty} \frac{1}{n} \log M(n, \varepsilon,\delta).
\end{align}

%%%%%%%%%%%%%%%%%%%%%%%%%%%%%%%%%%%%%%%%%%%%%%%%%%%%%%%%%%%%%%%%%%%
\begin{figure}[h] 
  \includegraphics[width=0.6\textwidth]{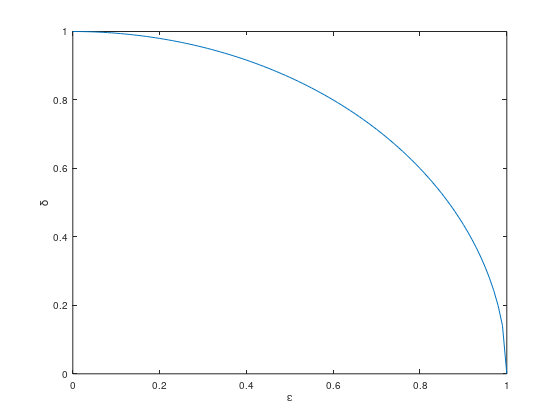}
  
  \caption{ {$\delta$ vs $\epsilon$: for any $\delta$ and $\epsilon$ satisfying $\delta \sqrt{1-\epsilon^2}+\epsilon \sqrt{1-\delta^2}<1$ the converse bound holds. }}%The  $\delta$ and $\epsilon$ satisfying }
  \label{fig: x function}
\end{figure}

\section{Strong Converse for the Private Capacity }\label{Sec:StrongConversePrivate}
%The definitions are similar to the pretty strong converse paper. $\rho^{X'XE^n}$ is the ``ideal'' final state on the reference $X'$, the output $X$ and the environment $E^n$. $\xi^{X'XE^n}$ is the decoded state. $\sigma^{X'B^nE^n}$ is the output of the channel.
In this section, we obtain a strong converse bound for the private classical capacity of anti-degradable channels 
(the private classical capacity of anti-degradable channels is equal to 0).%
\begin{theorem}
Let $\cN:A' \to B$ be an anti-degradable channel
with finite quantum systems $A'$ and $B$. Then, for any error $\epsilon$
and privacy $\delta$ satisfying $\delta \sqrt{1-\epsilon^2}+\epsilon \sqrt{1-\delta^2}<1$
 and every integer $n$,
\begin{align*}
    \log M(n,\epsilon,\delta)&\leq   2\log (\frac{1}{\cos(\alpha+\beta)}),
\end{align*}
where $ \beta=\sin^{-1}(\delta)$ and $ \alpha=\sin^{-1}(\epsilon)$.
\end{theorem}
\begin{proof}
For the target state  $\rho^{ E^n X'} = \omega^{E^n} \ox \sum_m \frac{1}{M} \proj{m}^{X'}  $, the environment is decoupled from the classical 
reference, hence $H_{\min}(X'|E^n)_{\rho}=\log M$. We apply this to obtain
\begin{align}\label{eq: C_P_1}
    \log M(n,\epsilon,\delta)=H_{\min}(X'|E^n)_{\rho}
    \leq H_{\min}^{\delta}(X'|E^n)_{\sigma} \leq H_{\min}^{\delta}(X'|B^n)_{\sigma},
\end{align}
where the first inequality follows from the definition of smooth min-entropy and that
$P(\rho^{E^nX'},\sigma^{E^nX'})\leq \delta$ holds. The second inequality is due to data processing: for anti-degradable channels $B^n$ can be obtained  
by applying a CPTP map on $E^n$.
Also, we note that for the target decoded state $\rho^{X'X}$ in Eq.~(\ref{eq: ideal rho^XX'}) we have $H_{\max}(X'|X)_{\rho}=0$:
\begin{align}\label{eq: C_P_2}
    0=-H_{\max}(X'|X)_{\rho} 
    \leq -H_{\max}^{\epsilon}(X'|X)_{\xi} \leq -H_{\max}^{\epsilon}(X'|B^n)_{\sigma}
\end{align}
where the first inequality follows from the definition of smooth max-entropy and that
$P(\rho^{X'X},\xi^{X'X})\leq \epsilon$ holds. 
The last inequality is data processing: the output on $X$ is obtained by a applying a CPTP on $B^n$.
Let $ \beta:=\sin^{-1}(\delta)$ and $ \alpha:=\sin^{-1}(\epsilon)$.
By adding Eq. (\ref{eq: C_P_1}) and Eq. (\ref{eq: C_P_2}) we obtain
\begin{align}
    \log M(n,\epsilon,\delta)&\leq H_{\min}^{\sin \beta}(X'|B^n)_{\sigma}-H_{\max}^{\sin \alpha}(X'|B^n)_{\sigma} \\
    &\leq  H_{\max}^{\sin \alpha}(X'|B^n)_{\sigma}-H_{\max}^{\sin \alpha}(X'|B^n)_{\sigma}+
    \log (\frac{1}{\cos^2(\alpha+\beta)})\\
    &=  2\log (\frac{1}{\cos(\alpha+\beta)}),
\end{align}
where the second line follows from Lemma~\ref{lemma: H_min < H_max+terms}, %and it  holds for $2\alpha<\frac{\pi}{2}$ implying that the bound holds for any $\epsilon <\frac{1}{\sqrt{2}}$
and it holds for $\alpha+\beta <\frac{\pi}{2}$. This implies that for any $\delta$ and $\epsilon$ satisfying $\delta \sqrt{1-\epsilon^2}+\epsilon \sqrt{1-\delta^2}<1$ the strong converse holds. This relation is illustrated in Fig.~1. 

\end{proof}

%%%%%%%%%%%%%%%%%%%%%%%%%%%%%%%%%%%%%%%%%%%%%%%%%%%%%%%%%%%%%%%%%

\section{Pretty strong converse for the quantum capacity}\label{Sec:StrongConverseQuantum}
%The definitions are similar to the pretty strong converse paper. 
In this section, we provide a simple proof that a pretty strong converse bound holds for the quantum capacity of anti-degradable channels
(the quantum capacity of anti-degradable channels is equal to 0).%
\begin{theorem}
Let $\cN:A' \to B$ be an anti-degradable channel
with finite quantum systems $A'$ and $B$. Then, for any error $\epsilon< \frac{1}{\sqrt{2}}$
and every integer $n$,
\begin{align*}
   \log N(n,\epsilon)\leq  \log (\frac{1}{\cos(2\alpha)}). 
\end{align*}
where $ \alpha=\sin^{-1}(\epsilon)$.
\end{theorem}

\begin{proof}
%
%$\Phi^{RA}$ is the ``ideal'' final state. $\xi^{RA}$ is the decoded state. $\sigma^{RB^nE^n}$ is the output of the channel. 
%
We start by using the fact that for a maximally entangled state $\Phi^{RA}$ of dimension $N$, $H_{\max}(R|A)_{\Phi}=-\log N$, namely 
\begin{align}\label{eq: C_Q_1}
    \log N(n,\epsilon)=-H_{\max}(R|A)_{\Phi}
    \leq -H_{\max}^{\epsilon}(R|A)_{\xi} \leq -H_{\max}^{\epsilon}(R|B^n)_{\sigma}
\end{align}
where the first inequality follows from the definition of smooth max-entropy and that
$P(\Phi^{RA},\xi^{RA})\leq \epsilon$ holds. 
The second inequality is due to data processing. 
From the above inequality and the duality relation $-H_{\max}^{\epsilon}(R|B^n)_{\sigma}=H_{\min}^{\epsilon}(R|E^n)_{\sigma}$ of Eq.~(\ref{eq: smooth duality}) we obtain
\begin{align}\label{eq: C_Q_2}
    \log N(n,\epsilon)\leq H_{\min}^{\epsilon}(R|E^n)_{\sigma} \leq H_{\min}^{\epsilon}(R|B^n)_{\sigma},
\end{align}
 where the second inequality is due to data processing: $B^n$ can be obtained  
by applying a CPTP map on $E^n$.
Let $ \alpha:=\sin^{-1}(\epsilon)$.
By adding Eq. (\ref{eq: C_Q_1}) and Eq. (\ref{eq: C_Q_2}) we obtain
\begin{align}\label{eq: C_Q_3}
    2\log N(n,\epsilon)&\leq H_{\min}^{\sin \alpha}(R|B^n)_{\sigma}-H_{\max}^{\sin \alpha}(R|B^n)_{\sigma}\nonumber\\
    & \leq H_{\max}^{\sin \alpha}(R|B^n)_{\sigma}-H_{\max}^{\sin \alpha}(R|B^n)_{\sigma}+\log (\frac{1}{\cos^2(2\alpha)}) \nonumber\\
    &=  \log (\frac{1}{\cos^2(2\alpha)}). 
\end{align}
where the second line follows from Lemma~\ref{lemma: H_min < H_max+terms}, %and it  holds for $2\alpha<\frac{\pi}{2}$ implying that the bound holds for any $\epsilon <\frac{1}{\sqrt{2}}$
and it holds for $2\alpha <\frac{\pi}{2}$. This implies that for any $\epsilon < \frac{1}{\sqrt{2}}$  the pretty strong converse holds. 
\end{proof}

\section{Discussion}\label{Sec:Discussion}

In this work we report progress on the problem of proving a strong converse for the private and quantum capacity. In particular we prove a strong converse for the former and a pretty strong converse for the later for the particular choice of anti-degradable channels. 
While these go beyond previous results, numerous open questions remain. Our proof strategy appears to be particular to the chosen class of channels. Proving something similar for degradable channels would constitute major progress in the field. 
%Why cannot we get strong converse 
%We can cancel out terms in deg is not possible 

\bigskip
\noindent \textbf{Acknowledgments.}
%I am grateful to  Andreas Winter and Debbie Leung for helpful discussions and comments on this manuscript. 
%Robert K\"onig for helpful discussions and comments on this manuscript and to Chokri Manai for discussions regarding continuity issues in Lemma~\ref{lemma: g_m(alpha)}. 
%I also thank Kohdai Kuroiwa for pointing out an error in Lemma 14 in the first version of the paper.
%
%
We are grateful to Andreas Winter for his time and many helpful discussions taking place
at the ``Workshop on Information Theory and Related Fields: In Memory of Ning Cai'' in Bielefeld, Germany. We extend
our sincere gratitude to the organizers of that event, particularly Christian Deppe, for their warm hospitality
and gracious invitation. 
We would also like to thank Patrick Hayden for his valuable comments on the draft.
ZBK was supported by the Marie Sk{\l}odowska-Curie Actions (MSCA)
European Postdoctoral Fellowships (Project 101068785-QUARC) and 
and the Ada Lovelace Fellowship at Perimeter Institute for Theoretical Physics. CH received funding by the Deutsche Forschungsgemeinschaft (DFG, German
Research Foundation) – 550206990.

\bibliography{lib}

\end{document}